\documentclass[11pt]{article}
\usepackage[backend=biber, maxbibnames=99]{biblatex}

\usepackage{amsmath,amsfonts,mathrsfs,mathtools}
\usepackage{amsthm}
\usepackage[margin=1in]{geometry}
\usepackage{hyperref}
\hypersetup{colorlinks}
\usepackage{subcaption}
\usepackage{xfrac}
\usepackage{enumitem}
\usepackage{float}
\usepackage{booktabs}
\usepackage{multirow}
\usepackage{algorithm}
\usepackage{algorithmic}

\theoremstyle{plain}
\newtheorem{theorem}{Theorem}[section]
\newtheorem{lemma}[theorem]{Lemma}

\newtheorem{proposition}[theorem]{Proposition}

\theoremstyle{definition}
\newtheorem{definition}[theorem]{Definition}
\newtheorem{assumption}[theorem]{Assumption}
\newtheorem{observation}[theorem]{Observation}

\DeclareMathOperator*{\argmin}{arg\,min}
\DeclareMathOperator*{\poly}{poly}
\DeclareMathOperator*{\maxdist}{dist-slack}
\DeclareMathOperator*{\polylog}{polylog}

\DeclareMathOperator*{\support}{support}
\DeclarePairedDelimiter\ceil{\lceil}{\rceil}

\begin{document}

\title{Stochastic Optimization and Learning for Two-Stage Supplier Problems}

\author{Brian Brubach\thanks{Wellesley College.
Email: \href{mailto:bb100@wellesley.edu}{bb100@wellesley.edu}} \and Nathaniel Grammel\thanks{University of Maryland, College Park. 
Email: \href{mailto:ngrammel@cs.umd.edu}{ngrammel@umd.edu}} \and David G. Harris\thanks{University of Maryland, College Park. Email: \href{mailto:davidgharris29@gmail.com}{davidgharris29@gmail.com}} \and Aravind Srinivasan\thanks{University of Maryland, College Park. 
Email: \href{mailto:asriniv1@umd.edu}{asriniv1@umd.edu}} \and Leonidas Tsepenekas \thanks{University of Maryland, College Park. 
Email: \href{mailto:ltsepene@cs.umd.edu}{ltsepene@cs.umd.edu}} \and Anil Vullikanti \thanks{University of Virginia. Email: \href{mailto:vsakumar@virginia.edu}{vsakumar@virginia.edu}}}

\maketitle

\date{}

\begin{abstract}
   The main focus of this paper is radius-based (supplier) clustering in the two-stage stochastic setting with recourse, where the inherent stochasticity of the model comes in the form of a budget constraint. In addition to the standard (homogeneous) setting where all clients must be within a distance $R$ of the nearest facility, we provide results for the more general problem where the radius demands may be \emph{inhomogeneous} (i.e., different for each client). We also explore a number of variants where additional constraints are imposed on the first-stage decisions, specifically matroid and multi-knapsack constraints, and provide results for these settings. 

We derive results for the most general distributional setting, where there is only black-box access to the underlying distribution. To accomplish this, we first develop algorithms for the \emph{polynomial scenarios} setting; we then employ a novel \emph{scenario-discarding} variant of the standard \emph{Sample Average Approximation (SAA)} method, which crucially exploits properties of the restricted-case algorithms. We note that the scenario-discarding modification to the SAA method is necessary in order to optimize over the radius.

\end{abstract}

\textit{\textbf{Keywords --- }
     Clustering, Facility location, stochastic optimization, approximation algorithms.
}

\section{Introduction}\label{sec:intro}
Stochastic optimization, first introduced in the work of Beale \cite{Beale1955} and Dantzig \cite{Dantzig1995}, provides a way to model uncertainty in the realization of the input data. In this paper, we give approximation algorithms for a family of problems in stochastic optimization, and more precisely in the $2$-\emph{stage recourse model} \cite{Swamy2006}. 
Our formal problem definitions follow.

We are given a set of clients $\mathcal{C}$ and a set of facilities $\mathcal{F}$, in a metric space with a distance function $d$. We let $n = |\mathcal{C}|$ and $m= |\mathcal{F}|$. Our paradigm unfolds in two stages. First, in \emph{stage-I}, each $i \in \mathcal{F}$ has a cost $c^I_i$, but at that time we do not know which clients from $\mathcal{C}$ will need service. At this point, we can proactively open a set of facilities $F_I$. After committing to $F_I$, a \emph{scenario} $A$ is sampled from some underlying distribution $\mathcal D$, which specifies some subset of clients $\mathcal C^A$ needing service; each $i \in \mathcal{F}$ now has a cost $c^A_i$ (which may vary across scenarios $A \in \mathcal{D}$). When this scenario $A$ arrives we can augment the solution by opening some additional \emph{stage-II} facilities $F_A$.

We then ask for $F_I$ and $F_A$, such that every client $j \in \mathcal C^A$ has $d(j, F_I \cup F_A) \leq R_j$ for every $A$ that materializes, where $R_j$ is some given radius dovemand. (We define $d(j, S) = \min_{i \in S}d(i,j)$ for any $j\in \mathcal{C}$ and $S \subseteq \mathcal{F}$.)
Furthermore, the expected opening cost is required to be at most some given budget $B$, i.e., $\sum_{i \in F_I}c^{I}_i + \mathbb{E}_{A \sim \mathcal{D}}\Big{[}\sum_{i \in F_A}c^{A}_i\Big{]} \leq B$.  We call this problem \emph{Two-Stage Stochastic Supplier} or  \textbf{2S-Sup} for short. 

If the values $R_j$ are all equal to each other, we say the instance is \emph{homogeneous} and write simply $R$. There may be additional search steps to minimize the radii $R_j$ and/or the budget $B$. For homogeneous instances, one natural choice is to fix $B$ and choose the smallest radius $R = R^*$.

For brevity, we write $j \in A$ throughout as shorthand for $j \in \mathcal C^A$. Also, to avoid degenerate cases, we always assume that $d(j, \mathcal F) \leq R_j$ for all clients $j$.

 \subparagraph*{Additional Stage-I Constraints:} Beyond the basic version of the problem, we also consider variants with additional hard constraints on the set of chosen stage-I facilities.

In \emph{Two-Stage Stochastic Matroid Supplier} or \textbf{2S-MatSup} for short, the input also includes a matroid $\mathcal{M}$. In this case, we additionally require that $F_I$ be an independent set of $\mathcal{M}$. 

In \emph{Two-Stage Stochastic Multi-knapsack Supplier} or \textbf{2S-MuSup} for short, there are $L$ additional knapsack constraints on $F_I$. Specifically, we are given budgets $W_{\ell} \geq 0$ and integer weights $f^{\ell}_i$ for each $i \in \mathcal{F}$ and index $\ell \in [L]$, such that the stage-I facilities should satisfy $\sum_{i \in F_I}f^{\ell}_i \leq W_{\ell}$ for all $\ell$.  In this case we further define a parameter $\Lambda = \prod_{\ell=1}^{L} W_{\ell}$.

\subparagraph*{Representing the Distribution:}
The most general way to represent the scenario distribution $\mathcal D$ is the \emph{black-box} model \cite{Swamy2011, Gupta20004, Shmoys2006, Swamy2019, multi2005}, where we have access to an oracle to sample scenarios $A$ according to $\mathcal{D}$. We also consider the \emph{polynomial-scenarios} model \cite{Srini2007, Immorlica2004, Ravi2004, Gupta2004}, where the distribution $\mathcal D$ is listed explicitly. We use the suffixes \textbf{BB} and \textbf{Poly} to distinguish these settings.  For example, \textbf{2S-Sup-BB} is the previously defined \textbf{2S-Sup} in the black-box model.

In either case, we define $\support(\mathcal D)$ to be the set of scenarios with non-zero probability under $\mathcal D$. Thus, in the polynomial-scenarios model, $\support(\mathcal D)$ is provided as a finite list of scenarios $Q = \{A_1, \dots, A_N \}$ along with occurrence probabilities $p_{A_i}$.  For brevity, we also write $A \in \mathcal D$ as shorthand for $A \in \support(\mathcal D)$ and $|\mathcal D|$ for $|\support(\mathcal D)|$. Note that a scenario $A \in \mathcal{D}$ corresponds not only to a subset ${\mathcal{C}}^A$ of clients but also to the \emph{costs} $c^A_i \in \mathbb{R}$ for each facility; thus, $\support(\mathcal D) \subseteq 2^{\mathcal C} \times \mathbb{R}^\mathcal{F}$ and, in the black-box setting, $|\mathcal D|$ may be uncountably infinite.
 
 In both settings, our algorithms must have runtime polynomial in $n, m$. For the polynomial-scenarios case, the runtime should also be polynomial in the number of scenarios $N = |\mathcal D |$.
\subsection{Motivation}

Clustering is a fundamental task in unsupervised and self-supervised learning. The stochastic setting models situations in which decisions must be made in the presence of uncertainty and are of particular interest in learning and data science. The black-box model is motivated by data-driven applications where specific knowledge of the distribution is unknown but we have the ability to sample or simulate from the distribution.  To our knowledge, radius minimization has not been previously considered in the two-stage stochastic paradigm. Most prior work in this setting has focused on \emph{Facility Location} \cite{Srini2007, Swamy2011, Ravi2004, Shmoys2006, Pal11, Swamy2019, multi2005}. On similar lines, \cite{agrawal2008} studies a stochastic $k$-center variant, where points arrive independently and each point only needs to get covered with some given probability. \textbf{2S-Sup} is the natural two-stage counterpart of the well-known \textbf{Knapsack-Supplier} problem, which has a well-known $3$-approximation \cite{Hochbaum1986}. 

From a practical perspective, our problems are motivated from questions of healthcare resource allocation~\cite{devries-pom20}.
For instance, during the COVID-19 pandemic, testing and vaccination centers were deployed at different kinds of locations, and access was an important consideration~\cite{li:aamas22, mehrab2022data}; access can be quantified in terms of different objectives including distance, as in our work. Here,
$\mathcal{F}$ and $\mathcal{C}$ correspond to such locations and the population affected by the outbreak, and needing services, respectively.
$\mathcal{D}$ corresponds to a distribution of disease outcomes.
Health departments prepared some vaccination and testing sites in advance, based on projected demands~\cite{bertsimas2022locate}, i.e., in stage-I, which may have multiple benefits; for example, the necessary equipment and materials might be cheaper and easier to obtain. 

An outbreak is an instance from $\mathcal{D}$, and after it actually happened, additional testing and vaccination locations were deployed or altered based on the new requirements, e.g.,~\cite{mehrab2022data}, which corresponds to stage-II decisions.
To continue this example, there may be further constraints on $F_I$, irrespective of the stage-II decisions, which cannot be directly reduced to the budget $B$. For instance, there may be a limited number of personnel available prior to the disease outbreak, assuming that facility $i$ requires $f_i$ people to keep it operational during the waiting period. (These additional stage-I constraints have not been previously considered  in the two-stage stochastic regime.)

\subsection{Our Generalization Scheme and Comparison with Previous Results}\label{sec:Our-Gen}
Our main goal is to develop algorithms for the black-box setting. As usual in two-stage stochastic problems, this has three steps. First, we develop algorithms for the simpler polynomial-scenarios model. Second, we sample a small number of scenarios from the black-box oracle and use our polynomial-scenarios algorithms to (approximately) solve the problems on them. Finally, we extrapolate the solution to the original black-box problem. This overall methodology is called \emph{Sample Average Approximation (SAA)}.

Unfortunately, standard SAA approaches \cite{swamySAA, Charikar2005} do not directly apply to radius minimization problems. On a high level, the obstacle is that radius-minimization requires estimating the cost of each approximate solution; counter-intuitively, this may be harder than optimizing the cost (which is what is done in previous results). See Appendix~\ref{sec:appendix} for an in-depth discussion. 

We need a new SAA approach. Consider a homogeneous problem instance, where we have a specific guess for the radius $R$. We sample a set of $N$ scenarios $Q = \{A_1, \dots, A_N \}$ from the oracle. We then run our polynomial-scenarios $\eta$-approximation algorithms on $Q$, which are guaranteed to provide solutions covering each client within distance $\eta R$. If $R$ is guessed correctly, and $N$ is chosen appropriately, these solutions have cost at most $(1+\epsilon)B$ on $Q$ with good probability. In the end we keep the minimum guess for $R$ whose cost over the samples is at most $(1+\epsilon)B$.  

For this guess $R$, the polynomial-scenarios algorithm returns a stage-I set $F_I$, and a stage-II set $F_{A}$ for each $A \in Q$. \textbf{Our polynomial-scenarios algorithms are also designed to satisfy two additional key properties.} First, given $F_I$ and any $A \notin Q$, there is an \emph{efficient} process to \emph{extend} the algorithm's output to a stage-II solution $F_A$ with $d(j,F_I \cup F_A) \leq \eta R_j$ for all $j \in A$. Second, irrespective of $Q$, the set $\mathcal S$ of resulting black-box solutions has only exponential size as a function of $n$ and $m$ (by default, it could have size $2^{m |\mathcal D|}$, and $\mathcal D$ may be exponentially large or even infinite). \textbf{We call algorithms satisfying these properties \emph{efficiently generalizable}}.

After using the extension process to construct a solution for every $A$ that materializes, we use a final \emph{scenario-discarding} step. Specifically, for some choice of $\alpha \in (0,1)$, we first determine a threshold value $T$ corresponding to the $\ceil*{\alpha N}^{\text{th}}$ costliest scenario of $Q$. Then, if for an arriving $A$ the computed set $F_A$ has stage-II cost more than $T$, we perform no stage-II openings by setting $F_A = \emptyset$ (i.e., we ``give up" on $A$). This step coupled with the bounds on $|\mathcal S|$ ensure that the overall opening cost of our solution is at most $(1+\epsilon)B$. Note that discarding implies that there may exist some scenarios $A$ with $d(j, F_I \cup F_A) > \eta R_j$ for some $j \in A$, but these only occur with frequency $\alpha$.

In Section \ref{sec:Gen}, we formally present our SAA scheme. We summarize it as follows:
\begin{theorem}
\label{thm11}
Suppose we have an efficiently generalizable, $\eta$-approximation for the polynomial-scenarios variant of any of the problems we study. Let $\mathcal{S}$ be the set of all potential black-box solutions its extension process may produce. Then, for any $\gamma, \epsilon, \alpha \in (0,1)$ and with $O\big{(}\frac{1}{\epsilon \alpha}  \log \big{(}\frac{nm|\mathcal{S}|}{\gamma}\big{)} \log  \big{(} \frac{nm}{\gamma} \big{)} \big{)}$ samples, we can compute a black-box solution $F_I$, $F_A$ for all $A \in \mathcal{D}$ such that, with probability at least $1 - \gamma$, one of the two conditions holds: 
\begin{enumerate}
\item $\sum_{i \in F_I}c^{I}_i + \mathbb{E}_{A \sim \mathcal{D}}[\sum_{i \in F_A}c^{A}_i] \leq (1+\epsilon) B$ 
and $\Pr_{A \sim \mathcal{D}}[d(j, F_I \cup F_A) \leq \eta R_j] \geq 1 - \alpha \quad  \forall j \in A$.
\item The algorithm returns INFEASIBLE and the original problem instance was also infeasible.
\end{enumerate}
\end{theorem}

In particular, for homogeneous problem instances, we \emph{optimize} over the optimal radius:
\begin{theorem}
\label{thm11a}
Suppose we have an efficiently generalizable, $\eta$-approximation for the polynomial-scenarios variant of any of the problems we study. Let $\mathcal{S}$ be the set of all potential black-box solutions its extension process may produce. Then, for any $\gamma, \epsilon, \alpha \in (0,1)$ and with $O\big{(}\frac{1}{\epsilon \alpha}  \log \big{(}\frac{nm|\mathcal{S}|}{\gamma}\big{)} \log  \big{(} \frac{nm}{\gamma} \big{)} \big{)}$ samples, we can compute a radius $R \leq R^*$ and black-box solution $F_I$, $F_A$ for all $A \in \mathcal{D}$ such that, with probability at least $1 - \gamma$, we have $$
\sum_{i \in F_I}c^{I}_i + \mathbb{E}_{A \sim \mathcal{D}}[\sum_{i \in F_A}c^{A}_i] \leq (1+\epsilon)B, \quad \text{and} \quad \Pr_{A \sim \mathcal{D}}[d(j, F_I \cup F_A) \leq \eta R] \geq 1 - \alpha, ~\forall j \in A.
$$
\end{theorem}

This adaptive thresholding may be of independent interest; for instance, it might be able to improve the sample complexity in the SAA analysis of \cite{Charikar2005}. By contrast, simpler non-adaptive approaches (e.g., $T = \frac{B}{\alpha}$) would have worse dependence on $\alpha$ and $\epsilon$ ($\frac{1}{\epsilon^2 \alpha^2}$ vs  $\frac{1}{\epsilon \alpha}$ as we achieve). 

We remark that if we make an additional assumption that the stage-II cost is at most some polynomial value $\Delta$, we can use standard SAA techniques without discarding scenarios; see Theorem~\ref{thm13} for full details. However, this assumption is stronger than is usually used in the literature for two-stage stochastic optimization.

In later sections, we follow up with some efficiently generalizable algorithms for a variety of problems, which we summarize as follows:
\begin{theorem}\label{intr-thm2}
We obtain the following efficiently generalizable algorithms
\begin{itemize}
    \item $3$-approximation for homogeneous \textbf{2S-Sup-Poly} with $|\mathcal{S}| \leq (n+1)!$.
    
    \item $5$-approximation for homogeneous \textbf{2S-MatSup-Poly} with $|\mathcal{S}| \leq 2^m$. 
    
    \item $5$-approximation for homogeneous \textbf{2S-MuSup-Poly}, with $|\mathcal{S}| \leq 2^m$ and runtime $\poly(n,m,\Lambda)$. 
    
    \item $11$-approximation for inhomogeneous \textbf{2S-MatSup-Poly}, with $|\mathcal{S}| \leq 2^m$.
\end{itemize}
\end{theorem}

The $3$-approximation for \textbf{2S-Sup-Poly} is presented in Section \ref{sec:3-apprx}, based on a novel LP rounding technqiue; notably, its approximation ratio matches the lower bound of the non-stochastic counterpart (\textbf{Knapsack Supplier}). 

The other three results are based on a reduction to a single-stage, deterministic robust outliers problem described in Section \ref{sec:mu-apprx}; namely, convert any $\rho$-approximation algorithm for the robust outlier problem into a $(\rho+2)$-approximation algorithm for the corresponding two-stage stochastic problem. This is similar to a robust supplier problem considered in \cite{bajpai2021revisiting} under the name \emph{priority center}, and many of the approximation algorithms of \cite{bajpai2021revisiting} can be adapted to our setting.

We follow up with $3$-approximations for the homogeneous robust outlier $\textbf{MatSup}$ and $\textbf{MuSup}$ problems, which are slight variations on algorithms of \cite{Chakra} (specifically, our approach in Section~\ref{sec:homogmumat} is a variation on their solve-or-cut methods). In Section~\ref{sec:5-apprx}, we describe a 9-approximation algorithm for an inhomogeneous $\textbf{MatSup}$ problem, which is an extension of results in \cite{Harris2018} and \cite{bajpai2021revisiting} (specifically, our method in Algorithm~\ref{alg-4} involves an iterative rounding approach similar to that of~\cite{Harris2018}). This new algorithm is intricate and may be of interest on its own.

\subparagraph*{Remark:} With our polynomial-scenarios approximation algorithms, the sample bounds of Theorem~\ref{thm11} for homogeneous \textbf{2S-Sup}, \textbf{2S-MatSup} and \textbf{2S-MuSup} instances are $\tilde O(\frac{ n  }{\epsilon \alpha}), \tilde O( \frac{m}{\epsilon \alpha})$ and $\tilde O(\frac{ m }{\epsilon \alpha})$ respectively. (Here, $\tilde O()$ hides $\polylog(n,m,1/\gamma)$ factors.)

There is an important connection between our generalization scheme and the design of our polynomial-scenarios approximation algorithms. In Theorem~\ref{thm11}, the sample bounds are given in terms of the \emph{cardinality} $|\mathcal S|$.  Our polynomial-scenarios algorithms are carefully designed to make $|\mathcal S|$ as small as possible. \emph{Indeed, one of the major contributions of this work is to show that effective bounds on $|\mathcal S|$  are possible for sophisticated approximation algorithms using complex LP rounding.}

 Following the lines of \cite{swamySAA}, it may be possible to replace this dependence with a notion of dimension of the underlying convex program. However, such general bounds would lead to \emph{significantly} larger complexities, consisting of very high order polynomials of $n$, $m$.

\subsection{Notation and Important Subroutines}
For $k \in \mathbb{N}$, we let $[k]$ denote $\{1,2,\hdots, k\}$. For a vector $\alpha = (\alpha_1, \alpha_2, \hdots, \alpha_k)$ and a subset $X \subseteq [k]$, we write $\alpha(X) = \sum_{i \in X}\alpha_i$.  For a client $j$ and radius demand $R_{j} \geq 0$, we define $G_{j} = \{ i \in \mathcal F: d(i,j) \leq R_{j} \}$. We also write $i^I_{j} = \argmin_{i \in G_{j}} c^I_i$ and $i^A_{j} = \argmin_{i \in G_{j}} c^A_i$ for any scenario $A$ and $j \in A$.    (Here, the radius $R_j$ is assumed from context.)

We repeatedly use a key subroutine named GreedyCluster(), shown in Algorithm~\ref{alg-1}. Its input is a set of clients $\mathcal{Q}$ with target radii $\mathbf{R} = (R_j)_{j\in\mathcal{Q}}$, and an ordering function $g: \mathcal{Q} \mapsto \mathbb{R}$. Its output is a set $H \subseteq \mathcal{Q}$ along with a mapping $\pi: \mathcal{Q} \mapsto H$; here, the intent is that $H$ should serve as a set of disjoint ``representatives" for the full set $\mathcal Q$. For readability, we write $\pi j$ as shorthand for $\pi(j)$.
 
\begin{algorithm}[H]
\begin{algorithmic}
\STATE {$H \gets \emptyset$}

\FOR {each $j \in \mathcal Q$ in decreasing order of $g(j)$ (breaking ties by index number)}
\STATE {$H \gets H \cup \{ j \}$}
\FOR {each $j' \in \mathcal{Q}$ with $G_{j} \cap G_{j'} \neq \emptyset$}
\STATE {$\pi j' \gets j, \mathcal{Q} \gets \mathcal{Q} \setminus \{j' \}$}
\ENDFOR
\ENDFOR
\RETURN $(H, \pi)$
\end{algorithmic}
\caption{$\text{GreedyCluster}(\mathcal{Q}, \mathbf{R}, g)$}\label{alg-1}
\end{algorithm}

\begin{observation}\label{filtering}
For $(H, \pi) = \text{GreedyCluster}(\mathcal{Q}, \mathbf{R}, g)$, the following two properties hold: (i) for all distinct pairs $j,j' \in H$, we have $G_{j} \cap G_{j'} = \emptyset$; and (ii)  for all $j \in \mathcal{Q}$ we have $G_{j} \cap G_{\pi j} \neq \emptyset$, $d(j,\pi j)\leq 2R_{j}$, and $g(\pi j) \geq g(j)$.
\end{observation}

\section{Generalizing to the Black-Box Setting}
\label{sec:Gen}

Let us consider a two-stage problem $\mathcal{P}$, with polynomial-scenarios variant $\mathcal{P}$-\textbf{Poly} and black-box variant $\mathcal{P}$-\textbf{BB}. We denote a problem instance  by the tuple $\mathfrak I = (\mathcal{C}, \mathcal{F},\mathcal{M}_I, \mathcal D, B, \mathbf{R})$, where $\mathcal{C}$ is the set of clients, $\mathcal{F}$ the set of facilities $i$, each with stage-I cost $c^I_i$, $\mathcal{M}_I \subseteq 2^{\mathcal{F}}$ the set of legal stage-I openings (representing the stage-I specific constraints of $\mathcal{P}$), $\mathcal D$ is the distribution (provided either explicitly or as an oracle), $B$ the budget, and $\mathbf{R}$ the vector of radius demands.

We suppose we also have an $\eta$-approximation algorithm $\text{Alg}\mathcal{P}$ for $\mathcal{P}$-\textbf{Poly}.
\begin{definition}\label{str-def}
We define a \emph{strategy} $s$ to be a $(|\mathcal D|+1)$-tuple of facility sets $(F^s_I, F^s_{A})$, where $F^s_I \in \mathcal M_I$ (i.e. it is a feasible stage-I solution) and where $A$ ranges over $\support(\mathcal D)$. The set $F^s_I$ represents the facilities opened in stage-I, and  $F^s_A$ denotes the facilities opened in stage-II, when the arriving scenario is $A$.   The strategy may be listed explicitly (if $\mathcal D$ is listed explicitly, in the polynomial-scenarios model) or implicitly.

For a strategy $s$ and scenario $A$, we write $C^{II}(s,A)$ for $c^A(F^s_A)$ and $C^I(s)$ for $c^I(F_s^I)$ for brevity. We also define $C(s,A)= C^I(s) + C^{II}(s,A)$ and $\maxdist(s,A) = \max_{j \in A} d(j, F_I^s \cup F_A^s) / R_j$.
\end{definition}

\begin{assumption}\label{asm-1}
In the black-box model, we assume that for any strategy $s$ under consideration, the random variable $C(s,A)$ for $A \sim \mathcal{D}$ has a continuous CDF. This assumption is without loss of generality: we may simply add a dummy facility $i'$ whose stage-II cost is an infinitesimal smooth random variable, and which always gets opened in every scenario.  \emph{Note that this  assumption implies that for a finite set of scenarios $Q$, all values $C(s,A)$ for $A \in Q$ are distinct with probability one}.
\end{assumption}

\begin{definition} 
We say that the instance $\mathfrak I = (\mathcal{C}, \mathcal{F},\mathcal{M}_I, \mathcal D, B, R)$ is \emph{feasible} for $\mathcal{P}$ if there is a strategy $s^*$ satisfying:
$$
\mathbb E_{A \sim D}[ C(s^*,A) ] \leq B, \qquad \textstyle{\max_{A \in \mathcal D}} \maxdist(s^*, A) \leq 1
$$
We say that strategy $s^*$ is \emph{feasible} for instance $\mathfrak{I}$.
\end{definition}

\begin{definition}\label{poly-feas}
An algorithm $\text{Alg}\mathcal{P}$ is an \emph{$\eta$-approximation algorithm} for $\mathcal{P}$-\textbf{Poly}, if given any problem instance $\mathfrak J = (\mathcal{C}, \mathcal{F}, \mathcal{M}_I, Q, \mathcal D, B, R)$, it satisfies the following conditions:
\begin{enumerate}[label=\textbf{A\arabic*}]
    \item Its runtime is polynomial in its input size, in particular, on the number of scenarios $N = |\mathcal D|$.
    \item \label{poly-appr-1} It either returns a strategy $s$ with $\mathbb E_{A \sim \mathcal D}[ C(s, A) ] \leq B$ and  $\textstyle{\max_{A \in \mathcal D}} \maxdist(s, A) \leq \eta$, or the instance in infeasible and it returns INFEASIBLE.
\end{enumerate}
\end{definition}

(Note that if the original instance is infeasible, it is possible, and allowed,  for $\text{Alg}\mathcal{P}$ to return a strategy $s$ satisfying the conditions of \ref{poly-appr-1}.)

\begin{definition}\label{eff-gen-def}
An $\eta$-approximation algorithm $\text{Alg}\mathcal{P}$ for $\mathcal{P}$-\textbf{Poly} is \emph{efficiently generalizable}, if for every instance $\mathfrak J = (\mathcal{C}, \mathcal{F}, \mathcal{M}_I, \mathcal D', B, R)$ and for every distribution $\mathcal D'$ with $\support(\mathcal D') \subseteq \support(\mathcal D)$, there is an efficient procedure to extend $s$ to a strategy $\bar{s}$ on $\mathcal D$, satisfying:
\begin{enumerate}[label=\textbf{S\arabic*}]
    \item \label{ef-gen-2} $\textstyle{\max_{A \in \mathcal D}} \maxdist(\bar s, A) \leq \eta$.
    
    \item \label{ef-gen-1} Strategies $s$ and $\bar s$ agree on $\mathcal D'$, i.e. $F^{\bar{s}}_I = F^s_I$ and $F^{\bar{s}}_{A} = F^s_{A}$ for every $A \in \support(\mathcal D')$.
    
    \item \label{ef-gen-3} The set $\mathcal S$  of possible strategies $\bar s$ that can be achieved from any such $\mathcal D'$ satisfies $|\mathcal{S}| \leq \psi$, where $\psi$ is some fixed known quantity (irrespective of $\mathfrak J$ itself).
\end{enumerate}
\end{definition}

\subsection{Overview of the Generalization Scheme}
We next describe our generalization scheme based on discarding scenarios. As a warm-up exercise, let us first describe how standard SAA results could be applied \emph{if we have bounds on the stage-II costs}.

\begin{theorem}\label{thm13}
Suppose the stage-II cost is at most $\Delta$ for any scenario and any strategy, and suppose we have an efficiently generalizable, $\eta$-approximation algorithm $\text{Alg}\mathcal P$, and let $\gamma, \epsilon \in (0,1)$. Then there is an algorithm that uses $N = O \bigl( \frac{\Delta}{\epsilon^2}  \log \bigl( \frac{ |\mathcal S|}{\gamma} \bigr) \bigr)$ samples from $\mathcal D$ and, with probability at least $1 - \gamma$, it satisfies one of the following two conditions:  
\begin{enumerate}
\item[(i)] it outputs a strategy $\bar s$ with $\mathbb E_{\mathcal A \sim \mathcal D}[ C( \bar s, A) ] \leq (1+\epsilon) B$ and $\textstyle{\max_{A \in \mathcal D}} \maxdist(s, A) \leq \eta$.
\item[(ii)] it outputs INFEASIBLE and $\mathfrak I$ is also infeasible for $\mathcal{P}$-\textbf{BB}.
\end{enumerate}
\end{theorem}
\begin{proof}[Proof (Sketch)] Sample $N$ scenarios $Q = \{A_1, \dots, A_N \}$ from $\mathcal D$, and let $\mathcal D'$ be the empirical distribution. For any fixed strategy, the average cost of $s$ on $\mathcal D'$, i.e. $\tfrac{1}{N} C(s,A_i)$, is a sum of independent random variables with range at most $\Delta$. By standard concentration arguments, for $N = \Omega( \frac{ \Delta}{\epsilon^2} \log(1/\delta) )$, this will be within a $(1 \pm \epsilon/3)$ factor of its underlying expected cost, with probability at least $1 - \delta$. In particular, for the stated value $N$, there is a probability of at least $1 - \gamma$ that this holds for every strategy $s \in \mathcal S$, as well as for some feasible strategy $s^*$ for $\mathcal{P}$-\textbf{BB}.

Since $s^*$ has cost at most $(1+\epsilon/3) B$ on $\mathcal D'$, we can run $\text{Alg}\mathcal{P}$ with a budget $(1+\epsilon/3) B$, and it will return a strategy $s$ with cost $(1+\epsilon/3) B$ on $\mathcal D'$. This can be extended to a strategy $\bar s \in \mathcal S$, whose cost on $\mathcal D'$ is at most $(1+\epsilon/3)$ times its expected cost on $\mathcal D$, i.e. the cost of $\bar s$ is at most $(1+\epsilon/3)^2 B \leq (1+\epsilon) B$.
\end{proof}

We wish to avoid to avoid this dependence on $\Delta$. The first step of our generalization method, as in Theorem~\ref{thm13}, is to sample a set $Q$ of scenarios from $\mathcal D$, and then apply the $\text{Alg}\mathcal{P}$ on the empirical distribution, i.e. the uniform distribution on $Q$. The extension procedure of $\text{Alg}\mathcal{P}$ yields a strategy $\bar s$ for the entire distribution $\mathcal D$. Instead of using this strategy directly, we then modify it as follows: based on sample set $Q$, we find a threshold $T$ so that $C^{II}(s, A) > T$ for exactly $\alpha N$ samples. If a newly-arriving  $A$ has $C^{II}(\bar s, A)  > T$, we perform no stage-II opening. We let $\hat s$ denote this modified strategy, i.e. $F^{\hat s}_I = F^{\bar s}_I$ and $F^{\hat s}_A = \emptyset$ when $C^{II}(\bar s, A) > T$, otherwise $F^{\hat s}_A = F^{\bar s}_A$.

See Algorithm~\ref{alg-6} for details:

\begin{algorithm}[H]
\begin{algorithmic}
\REQUIRE {Parameters $\epsilon, \gamma, \alpha \in (0,1)$, $N \geq 1$ and a $\mathcal{P}$-\textbf{BB} instance $\mathfrak I = (\mathcal{C}, \mathcal{F},\mathcal{M}_I, \mathcal D, B, R)$.}
\FOR {$h=1, \dots, \lceil \log_{13/12} (1/\gamma) \rceil$}
\STATE {Draw $N$ independent samples from the oracle, obtaining set $Q = \{S_1, \hdots, S_N\}$}
\STATE {Set $\mathcal D'$ to be the uniform distribution on $Q$, i.e. each scenario $S_i$ has probability $1/N$}
\IF {$\text{Alg}\mathcal{P}(\mathcal{C}, \mathcal{F},\mathcal{M}_I, \mathcal D', (1+\epsilon)B, R)$ returns strategy $s$}
\STATE {Let $T$ be the $\lceil \alpha N \rceil^{\text{th}}$ largest value of $C^{II}(s,A)$ among all scenarios $A \in Q$}
\STATE {Use generalization procedure to obtain strategy $\bar s$ from $s$}
\STATE {Form strategy $\hat s$ by discarding scenarios $A$ with $C^{II}(s,A) > T$}
\RETURN $\hat s$\; 
\ENDIF
\ENDFOR
\RETURN INFEASIBLE
\end{algorithmic}
\caption{SAA Method for $\mathcal{P}$-\textbf{BB}.}\label{alg-6}
\end{algorithm}

In the following analysis, we will show that, with probability $1 - O(\gamma)$, the resulting strategy $\hat{s}$ has two desirable properties:  (1) its expected cost (taken over all scenarios in
$\mathcal{D}$) is at most $(1+2\epsilon)B$, and (2) it covers at least a $1-2\alpha$ fraction
(of probability mass) of scenarios of $\mathcal{D}$ (where by ``cover'' we mean
that all clients have an open facility within distance $\eta R_j$).

\subsection{Analysis of the Generalization Procedure}

The analysis will use a number of standard variants of Chernoff's bound. It also uses a somewhat more obscure concentration inequality of Feige \cite{feige}, which we quote here:
\begin{theorem}[\cite{feige}]
\label{feige-thm}
Let $X_1, \dots, X_n$ be nonnegative independent
random variables, with expectations $\mu_i = \mathbb E[X_i]$, and let $X = \sum_i X_i$. Then, for any $\delta > 0$, there holds
$$
\Pr( X < \mathbb E[X] + \delta) \geq \min \{ \delta/(1+\delta), 1/13 \}
$$
\end{theorem}

We now begin the algorithm analysis.
\begin{lemma}\label{s6-lem-feas}
If instance $\mathfrak I$ is feasible for $\mathcal{P}$-\textbf{BB} and $N \geq 1/\epsilon$, then Algorithm~\ref{alg-6} terminates with INFEASIBLE with probability at most $\gamma$.
\end{lemma}

\begin{proof}
By rescaling, we assume wlog that $B=1$. If $\mathfrak I$ is feasible, there exists some feasible strategy $s^\star$. Now, for any specific iteration $h$ in Algorithm \ref{alg-6}, let $X_v$ be the total cost of $s^{\star}$ on sample $S_v$. The random variables $X_v$ are independent, and the average cost of $s^{\star}$ on $\mathcal D'$ is $Y = \frac{1}{N}\sum^N_{v=1} X_v$. As $s^\star$ is feasible for $\mathfrak I$ we have $\mathbb{E}[X_v] = \mathbb E_{A \sim \mathcal D} [C(s^{\star},A)] \leq 1$. By Theorem~\ref{feige-thm}, we thus have:
\begin{align}
\Pr\Big{[} \sum^N_{v=1}X_v < \mathbb{E}\big{[}\sum^N_{v=1}X_v\big{]} + \epsilon N \Big{]} \geq \min \big{\{}\frac{\epsilon  N}{1 + \epsilon N}, \frac{1}{13} \big{\}}  \notag
\end{align}
When $N \geq 1/\epsilon$, we have $\epsilon N /(1 + \epsilon N) \geq 1/13$. Hence, with probability at least $1/13$, we get:
\begin{align*}
Y \leq \frac{1}{N}\sum^N_{v=1}\mathbb{E}[X_v] + \epsilon B \leq \mathbb E_{A \sim \mathcal D}[C(s^{\star}, A)]  + \epsilon B \leq (1+\epsilon)B
\end{align*}

When this occurs, then strategy $s^\star$ witnesses that $(\mathcal{C}, \mathcal{F},\mathcal{M}_I, \mathcal D', (1+\epsilon)B, R)$ is feasible for $\mathcal{P}$. Since $\text{Alg}\mathcal{P}$ is an $\eta$-approximation algorithm, the iteration terminates successfully. Repeating for $\lceil \log_{13/12}(1/\gamma) \rceil$ iterations brings the error probability down to at most $\gamma$. 
\end{proof}

\begin{proposition}\label{probE}
For any $\gamma,\alpha \in (0,1)$ and $N = O\big{(}\frac{1}{\alpha}  \log (\frac{\psi}{\gamma}) \big{)}$, there is a probability of at least $1 - \gamma$ that the algorithm outputs either INFEASIBLE, or the strategy $\bar s$ and threshold $T$ satisfy
\[
\Pr_{A \sim D} ( C^{II}(\bar s,A) > T) \leq 2 \alpha
\]
\end{proposition}
\begin{proof}
For any strategy $s$, let $t_s$ be the threshold  value $t_{s}$ such that $\Pr_{A \sim \mathcal D}( C^{II}(s,A) > t_{s}) = 2 \alpha$; this is well-defined by Assumption~\ref{asm-1}. For a strategy $s$ and iteration $h$ of Algorithm~\ref{alg-6}, let $E_{h,s}$ denote the event that $Q$ contains at most $\alpha N$ samples with stage-II cost exceeding $t_s$. We claim that, if no such event $E_{h,s}$ occurs, then we get the desired property in the proposition.  For, suppose the chosen strategy has $\Pr_{A \sim \mathcal D} ( C^{II}(\bar s,A) > T) > 2 \alpha$. So necessarily $T \leq t_{\bar s}$. Since $Q$ contains $\alpha N$ samples with cost exceeding $T$, it also contains at most $\alpha N$ samples with stage-II cost exceeding $t_{\bar s}$, and hence bad-event $E_{h, \bar s}$ has occurred.

Now let us bound the probability of $E_{h,s}$ for any fixed $h, s$. Let $X$ denote the number of samples with stage-II cost exceeding $t_s$; note that $X$ is a Binomial random variable with $N$ trials and with rate exactly $2 \alpha$. We can set $N = O( \frac{1}{\alpha} \log (\psi/\gamma))$ so that:
\begin{align*}
  \Pr(X \le \alpha N)
  = \Pr(X - \mathbb{E}(X) \le -\alpha N)
  \le e^{-(\alpha N)^{2} / 2\mathbb{E}(X)} =
  e^{-\alpha N / 4}
  = \frac{\gamma} {\psi \ceil{\log_{13/12}(\gamma)}}
\end{align*}
(for a reference on this Chernoff bound, see e.g.\ Theorem~A.1.13 of~\cite{alonSpencerProbMethod}).

To finish, take a union bound over the set of strategies $\mathcal S$, which has size at most $\psi$, and the total number of iterations $h$, which is at most $O(\log(1/\gamma)$.
\end{proof}

\begin{theorem}\label{failure}
For $\epsilon, \gamma,\alpha \in (0,1/8)$ and $N = O\big{(}\frac{1}{\epsilon \alpha}  \log (\frac{\psi}{\gamma}) \big{)}$, there is a probability of at least $1 - \gamma$ that the algorithm outputs either INFEASIBLE, or outputs a   strategy $\hat s$ with 
$$
\mathbb E_{A \sim \mathcal D}[C(\hat s,A)] \leq (1 + 2 \epsilon) B.$$
\end{theorem}
\begin{proof}
Wlog suppose $B = 1$. For any strategy $s$, let $t_s$ be the threshold  value $t_{s}$ such that  $\Pr_{A \sim \mathcal D}( c^A(s) > t_{s}) = \alpha/2$; again, this is well-defined by Assumption~\ref{asm-1}. We denote by $\phi(s)$ the modified strategy which discards all scenarios whose stage-II cost exceeds $t_{s}$. 

For an iteration $h$ and strategy $s$, let $E_{h,s}$  denote the bad-event that either (i) there are fewer than $\alpha N/4$ samples $A \in Q$ with $C^{II}(s,A) \geq t_{s}$ or (ii) there are more than $\alpha N$ samples $A \in Q$ with $C^{II}(s,A) \leq t_{s}$ or (iii) $\phi(s)$ has expected cost larger than $1 + 2 \epsilon$ and $\phi(s)$ has empirical cost at most $(1+\epsilon) - \alpha/4 \cdot t_{s}$.

We claim that the desired bound holds as long as no bad-event $E_{h,s}$ occurs. For, suppose that strategy $\bar s$ comes from some iteration $h$ of Algorithm~\ref{alg-6}, but the expected cost of $\hat s$ exceeds $1 + 2 \epsilon$. The cost of $\bar s$ on $Q$ is at most $1+\epsilon$ by assumption. If condition (ii) does not hold for $\bar s$, we must have $T \leq t_{\bar s}$, and hence $\phi(\bar s)$ also has expected cost larger than $1 + 2 \epsilon$.  Also, if condition (i) does not hold, there are at least $\alpha N/4$ samples with $C^{II}(\bar s, A) \geq t_{\bar s}$. These scenarios are discarded by $\phi(\bar s)$, so the empirical cost of $\phi(\bar s)$ on $Q$ is at most $y = (1+\epsilon) - \alpha/4 \cdot t_{s}$. Thus, condition (iii) holds.

We now turn to bounding the probability of $E_{h,s}$ for some fixed $s,h$. By an argument similar to Proposition~\ref{probE}, it can easily seen that events (i) and (ii) have probability at most $e^{-\Omega(N \alpha)}$. For (iii), let $Y = \frac{1}{N} \sum_{A \in Q} C(A,\phi(s))$ denote the empirical cost of $\phi(s)$ on $Q$. Here $Y$ is a sum of independent random variables, each of which is bounded in the range $[c^I(F_I^s), c^I(F_I^s) + t_{s}/N]$ and which has mean $\mu = \mathbb E_{A \sim \mathcal D} [ C(\phi(s), A) ]$. If condition (iii) holds, we have $\mu \geq 1+2 \epsilon$. By a variant of Chernoff's lower-tail bound (taking into account the range of the summands), we get the bound: 
$$
\Pr( Y \leq y ) \leq e^{ -\frac{N \mu (1 - y/\mu)^2}{2 t_{s}}}
$$

Since $\mu \geq x = (1 + 2 \epsilon) B$, and $y \leq x - (\epsilon B + \alpha t_s/4)$ and $x \in [1,2]$, monotonicity properties of Chernoff's bound imply
$$
\Pr( Y \leq y ) \leq \exp \Bigl( -\frac{N}{2 t_s} \cdot \frac{ (\epsilon + \alpha t_s/4)^2}{ 4 } \Bigr) 
$$

Simple calculus shows that this quantity is maximized at $t_s = 4 \epsilon/\alpha$, yielding the resulting bound $\Pr( Y \leq y ) \leq e^{-N \epsilon \alpha/8}$. So the overall probability of $E_{h,s}$ is at most $2 e^{-\Omega(N \alpha)} + e^{-\Omega(N \epsilon \alpha)}$. Taking $N = \Omega( \frac{ \log(\psi/\gamma) }{\epsilon \alpha})$ ensures that the total probability of all such bad-event is at most $\gamma$.
\end{proof}

In particular, Theorems~\ref{probE} and \ref{failure} together show Theorem~\ref{thm11}.  By optimizing over the radius, we also can show Theorem~\ref{thm11a}. 
\begin{proof}[Proof of Theorem~\ref{thm11a}]
Because $R^*$ is the distance between some facility and some client, it has at most $n m$ possible values. We run Algorithm~\ref{alg-6} for each  putative radius $R$, using error parameter $\gamma' = \frac{\gamma}{n m}$. We then return the smallest radius that did not yield INFEASIBLE, along with corresponding strategy $\hat s$. By a union bound over the choices for $R$, there is a probability of $1 - O(\gamma)$ that the following desirable events hold.  First, by Lemma~\ref{s6-lem-feas}, at least one iteration $R \leq R^*$ returns a strategy $\hat s_R$.  Also, by Proposition~\ref{probE} and Theorem~\ref{failure}, all choices of $R$ that return any strategy $\hat s_R$ have $\mathbb E_{A \sim \mathcal D}[C(\hat s_R, A)] \leq (1+\epsilon)B$ and $\Pr_{A \sim \mathcal{D}}[\maxdist(\hat s_R, A) > \eta] \leq \alpha$.

In particular, with probability $1 - O(\gamma)$, the smallest value $R$ that does not return INFEASIBLE has $R \leq R^*$ and satisfies these two conditions as well. The stated result holds by rescaling $\alpha, \epsilon, \gamma$. 

Note that we do not need fresh samples for each radius guess $R$; we can draw an appropriate number of samples $N$ upfront, and test all guesses in ``parallel'' with the same data.
\end{proof}

\textbf{We now describe some concrete approximation algorithms for $\mathcal{P}$-\textbf{Poly} problems. To emphasize that distribution $\mathcal D'$ is provided explicitly, we write $Q$ for the list of scenarios $A$, each with some given probability $p_A$. The solution strategy $s$ is simply a list of $F_I, F_A$ for  $A \in Q$. We also write just $F$ for the entire ensemble $F_I, F_A$, as well as $C(F,A) = c^I(F_I) + c^A(F_A$) for its cost on a scenario $A$.}

\section{Approximation Algorithm for Homogeneous \normalfont{2S-Sup}}
\label{sec:3-apprx}

In this section we tackle the simplest problem setting, designing an efficiently-generalizable $3$-approximation algorithm for homogeneous  \textbf{2S-Sup-Poly}. To begin, we are given a list of scenarios $Q$ together with their probabilities $p_A$, and a single target radius $R$.  Now consider LP (\ref{s2-LP-1})-(\ref{s2-LP-3}). 
\begin{align}
&\displaystyle\sum_{i \in \mathcal{F}}y^{I}_i \cdot c^{I}_i + \displaystyle\sum_{A \in Q} p_A \displaystyle\sum_{i \in \mathcal{F}}y^{A}_i \cdot c^{A}_i \leq B  \label{s2-LP-1}\\
&\displaystyle\sum_{i \in G_j}(y^{I}_i + y^{A}_i) \geq 1,~~ \forall j \in A \in Q  \label{s2-LP-2}\\
&0 \leq y^{I}_i, y^{A}_i \leq 1  \label{s2-LP-3}
\end{align}

Here (\ref{s2-LP-1}) captures the budget constraint, and (\ref{s2-LP-2}) captures the radius covering constraint. If the instance is feasible for the given \textbf{2S-Sup-Poly} instance, we can solve the LP. The rounding algorithm appears in Algorithm~\ref{alg-2}. 

\begin{algorithm}[H]
\begin{algorithmic}
\STATE {Solve LP (\ref{s2-LP-1})-(\ref{s2-LP-3}) to get a feasible solution $y^I, y^A: A \in Q$}
\IF {no feasible LP solution exists} \RETURN INFEASIBLE
\ENDIF
\STATE {$(H_I, \pi^I) \gets $ GreedyCluster$(\mathcal{C}, R, g^I)$, where $g^I(j) = y^I(G_j)$}
\FOR {each scenario $A \in Q$}
\STATE {$(H_A, \pi^A) \gets $ GreedyCluster$(A, R, g^A)$, where $g^A(j) = -y^I(G_{\pi^I j})$}

\ENDFOR
\STATE {Order the clients of $H_I$ as $j_1, j_2, \dots, j_h$ such that $y^I(G_{j_1}) \leq y^I(G_{j_2}) \leq \dots \leq y^I(G_{j_h})$}
\STATE {Create an additional ``dummy'' client $j_{h+1}$ with $y^I(G_{j_{h+1}}) > y^I(G_{j_{\ell}})$ for all $\ell \in [h]$}
\FOR {all integers $\ell = 1,2,\hdots, h+1$}
\STATE {$F_I \gets \{i^I_{j_k} ~|~ j_k \in H_I \text{ and } y^I(G_{j_k}) \geq y^I(G_{j_\ell})\}$}
\FOR {each $A \in Q$} 
\STATE {\textbf{do} $F_A \gets \{ i^A_{j} ~|~ j \in H_A \text{ and } F_I \cap G_{\pi^I j} = \emptyset \}$}
\ENDFOR
\IF {$\sum_{A \in Q} p_A C(F, A) \leq B$}
\RETURN ensemble $F$
\ENDIF

\ENDFOR
\RETURN INFEASIBLE
\end{algorithmic}
\caption{Correlated LP-Rounding Algorithm for \textbf{2S-Sup-Poly}}\label{alg-2}
\end{algorithm}

\begin{theorem}\label{s2-thm-1}
For any scenario $A \in Q$  and every $j \in A$, we have $d(j,F_I \cup F_A) \leq 3R$.
\end{theorem}
\begin{proof}
Recall that $d(j,\pi^I j)\leq 2R$ and $d(j,\pi^A j) \leq 2R$ for any $j \in A$. For  $j \in H_A$ the statement is clear because either $G_{\pi^I j} \cap F_I \neq \emptyset$ or $G_j \cap F_A \neq \emptyset$. So consider some $j \in A \setminus H_A$. If $G_{\pi^A j} \cap F_A \neq \emptyset$, then any facility $i \in G_{\pi^A j} \cap F_A$ will be within distance $3R$ from $j$. If on the other hand $G_{\pi^A j} \cap F_A = \emptyset$, then our algorithm guarantees $G_{\pi^I(\pi^A j)} \cap F_I \neq \emptyset$. Further, the stage-II greedy clustering yields $g^A(\pi^A j) \geq g^A(j)$ and so $y^I(G_{\pi^I j}) \geq y^I(G_{\pi^I \pi^A j})$. From the way we formed $F_I$ and the fact that $G_{\pi^I \pi^A j} \cap F_I \neq \emptyset$, we infer that $G_{\pi^I j} \cap F_I \neq \emptyset$ and hence $d(j, G_{\pi^I j} \cap F_I) \leq 3R$.
\end{proof}

\begin{theorem}\label{a-cost}
For a feasible instance, the algorithm does not return INFEASIBLE.
\end{theorem}

\begin{proof}
Consider the following random process to generate a solution: draw a random variable $\beta$ uniformly from $[0,1]$, and set $F^{\beta}_I = \{i^I_j ~|~ j \in H_I \text{ and } y^I(G_j) \geq \beta\}$, $F^{\beta}_A = \{i^A_j ~|~ j \in H_A \text{ and } F_I \cap G_{\pi^I j} = \emptyset\}$ for all $A \in Q$. For each possible draw for $\beta$, the resulting sets $F^{\beta}_I, F^{\beta}_A$ correspond to some iteration $\ell$ of the algorithm.  Hence, in order to show the existence of an iteration $\ell$ with $\sum_{A \in Q} p_A C(F,A) \leq B$, it suffices to show $\mathbb{E}_{\beta \sim [0,1]}[ \sum_A p_A C( F^{\beta}, A) ] \leq B$.

To start, consider some facility $i^I_j$ with $j \in H_I$ in stage-I. This is opened in stage I only if $\beta \leq y^I(G_j)$, and so $i^I_j$ is opened in stage-I with probability at most $y^I(G_j)$.  The sets $G_j$ are pairwise-disjoint for $j \in H_I$, so: 
\begin{align}
\mathbb{E}_{\beta \sim [0,1]}[c^I(F^{\beta}_I)] \leq \sum_{j \in H_I}c^I_{i^I_j} \cdot y^I(G_j) \leq  \sum_{i \in \mathcal{F}}y^{I}_i \cdot c^{I}_i \label{a-1}
\end{align}
Next, for any  $j \in H_A$ and any  $A \in Q$, the probability that $i^A_j$ is opened in stage-II for a given scenario $A$ is $1-\min(y^I(G_{\pi^I j}), 1)
\leq 1-\min(y^I(G_{j}),1) \leq y^A(G_j)$.
(The first inequality results from the greedy stage-I clustering that gives $y^I(G_{\pi^I j}) \geq y^I(G_{j})$, and the second from (\ref{s2-LP-2}).) Again, since the sets $G_j$ are pairwise-disjoint for $j \in H_A$, we get:  
\begin{align}
\mathbb{E}_{\beta \sim [0,1]}[c^A(F^{\beta}_A)] \leq \sum_{j \in H_A}c^A_{i^A_j}\cdot y^A(G_{j}) \leq \sum_{i \in \mathcal{F}}y^{A}_i \cdot c^{A}_i \label{a-2}
\end{align}
Combining (\ref{a-1}), (\ref{a-2}) and (\ref{s2-LP-1}) gives $
    \mathbb{E}_{\beta \sim [0,1]}[c^I(F^{\beta}_I)] + \sum_{A \in Q} p_A \cdot \mathbb{E}_{\beta \sim [0,1]}[c^A(F^{\beta}_A)] \leq B $.
\end{proof}

Finally, we can show that Algorithm~\ref{alg-2} fits the framework of Section \ref{sec:Gen}

\begin{theorem}
Algorithm~\ref{alg-2} is efficiently-generalizable, with $|\mathcal S| \leq (n+1)!$
\end{theorem}
\begin{proof}
The extension procedure is given in Algorithm \ref{alg-7} below; here we exploit the fact that stage-II decisions in Algorithm \ref{alg-2} only depend on information from the LP about stage-I variables.

\begin{algorithm}[H]
\begin{algorithmic}
\REQUIRE {Newly arriving scenario $A$}
\STATE {For every $j \in A$ set $g(j) \gets -y^I(G_{\pi^I j})$, where $y^I, \pi^I$ are as computed in Algorithm \ref{alg-2} }
\STATE {$(H_A, \pi^A) \gets $ GreedyCluster$(A, R, g)$}
\STATE {$F^{\bar s}_A \gets \{i^A_j ~|~ j \in H_A \text{ and } F_I \cap G_{\pi^I j} = \emptyset \}$}
\end{algorithmic}
\caption{Generalization Procedure for \textbf{2S-Sup-Poly}}\label{alg-7}
\end{algorithm}

Since Algorithm \ref{alg-7} mimics the stage-II actions of Algorithm \ref{alg-2}, it satisfies property \ref{ef-gen-1}. For \ref{ef-gen-2}, note that the proof of Theorem \ref{s2-thm-1} only depended on the fact that $d(j, \mathcal F) \leq R$ for all $j \in A$ (so that $d(j, \pi^I j) \leq 2 R$ and $d(j, \pi^A j) \leq 2 R)$; by our assumptions, this property holds for all $j \in \mathcal C$.

Finally, we claim that at most $(n+1)!$ strategies are achievable via Algorithm \ref{alg-7}. For, the constructed final strategy is determined by (1) the sorted order of $y^I(G_j)$ for all $j \in \mathcal{C}$, and (2) a minimum threshold $\ell'$ such that $G_{j_{\ell'}} \cap F_I \neq \emptyset$ with $j_{\ell'} \in H_I$. From these, we can determine the sets $H_I, F_I$ and $F_A, H_A$ for every $A \in \mathcal{D}$. There are $n!$ total possible orderings for the $y^I(G_j)$ values, and $n+1$ possible values for the threshold parameter $\ell'$ can take at most $n+1$ values.
\end{proof}

\section{A Generic Reduction to Robust Outlier Problems}
\label{sec:mu-apprx}
We now describe a generic method of transforming a given $\mathcal{P}$-\textbf{Poly} problem into a single-stage deterministic robust outlier problem. This will give us a 5-approximation algorithm for homogeneous \textbf{2S-MuSup} and \textbf{2S-MatSup} instances nearly for free; in the next section, we also use it obtain our 11-approximation algorithms for inhomogeneous \textbf{2S-MatSup}.

\paragraph*{Robust Weighted Supplier (RW-Sup):} In this model, we are given a set of clients $\mathcal{C}$ and a set of facilities $\mathcal{F}$, in a metric space with distance function $d$, where every client $j$ has a radius demand $R_j$. The input  also includes a weight $v_j \in \mathbb R_{\geq 0}$  for every client $j \in \mathcal{C}$, and a weight $w_i \in \mathbb R_{\geq 0}$ for every facility $i \in \mathcal{F}$. The goal is to choose a set of facilities $S \in \mathcal M$ such that
\begin{equation}
\label{rw-constraints}
\sum_{i \in S} w_i + \sum_{j \in \mathcal{C}: d(j,S) > R_j}v_j \leq V
\end{equation}
for a given budget $V$. Clients $j$ with $d(j,S) > R_j$ are called outliers. We say an algorithm $\text{AlgRW}$ is a $\rho$-approximation for an instance of $\textbf{RW-Sup}$ if it returns a solution $S \in \mathcal M$ with
\begin{equation}
\label{rw-constraints2}
\sum_{i \in S} w_i + \sum_{j \in \mathcal{C}: d(j,S) > \rho R_j}v_j \leq V
\end{equation}

The following is the fundamental reduction between the problems:
\begin{theorem}
\label{aga-thm0}
If we have a $\rho$-approximation algorithm for $\text{AlgRW}$ for given $\mathcal C, \mathcal F, \mathcal M, R$, then we can get an efficiently-generalizable $(\rho+2)$-approximation algorithm for the corresponding problem $\mathcal{P}$-\textbf{poly}, with $|\mathcal S| \leq 2^m$.
\end{theorem}

To show Theorem~\ref{aga-thm0}, we use the following Algorithm~\ref{alg-5} for a given set of provided scenarios $Q$ with probabilities $p_A$. 

\begin{algorithm}[H]
\begin{algorithmic}
\FOR {each scenario $A \in Q$}
\STATE {$(H_A, \pi^A) \gets $ GreedyCluster$(A, R, -R)$}

\ENDFOR
\STATE {Construct instance $\mathfrak{I}'$ of \textbf{RW-MuSup} with 
$$
\mathcal M = \mathcal M_I, \qquad V = B, \qquad w_i = c^I_i \text{ for each $i$},  \qquad v_j = \sum_{A: j \in H_A} p_A c_{i_j^A}^A \text{ for each $j$}
$$
}
\IF{$\text{AlgRW}(\mathfrak I')$ return feasible solution $F_I$}
\FOR {each scenario $A \in Q$}
\STATE {$F_A \gets \{ i^A_j  \mid \text{$j \in H_A$ with $d(j, F_I) > \rho R_j$} \}$}
\ENDFOR
 \RETURN ensemble $F$

 \ENDIF
\RETURN INFEASIBLE
\end{algorithmic}
\caption{Generic Approximation Algorithm for \textbf{2S-Sup-Poly}}\label{alg-5}
\end{algorithm}

\begin{lemma}\label{mu-feas}
If the original instance $\mathfrak I$ is feasible, then \textbf{RW-Sup} instance $\mathfrak I'$ is also feasible.
\end{lemma}

\begin{proof}
Consider some feasible solution $F^{\star}$ for $\mathcal{P}$-\textbf{poly}. We claim that $F^{\star}_I$ is a valid solution for $\mathfrak I'$. It clearly satisfies $F^{\star}_I \in \mathcal M$. Now, for any $A \in Q$, any client $j \in H_A$ with $d(j,F^{\star}_I) > R_j$ must be covered by some facility $x^A_j \in G_j \cap F^{\star}_A$. Since $F^*$ is feasible, and the sets $G_{j}$ are pairwise disjoint for $j \in H_A$, we have:
\begin{align}
\sum_A p_A \sum_{i \in F^{\star}_A} c^A_i \geq \sum_A p_A \sum_{\substack{j \in H_A : \\ d(j,F^{\star}_I) > R_j}} c^A_{x^A_j} \geq \sum_A p_A \sum_{\substack{j \in H_A : \\ d(j,F^{\star}_I) > R_j}} c^A_{i^A_j} = \sum_{\substack{j \in \mathcal{C} : \\ d(j,F^{\star}_I) > R_j}}v_j\notag
\end{align}

Thus, $$
\sum_{i \in F^{\star}_I} w_i + \sum_{\substack{j \in \mathcal{C} : \\ d(j,F^{\star}_I) > R_j}}v_j \leq \sum_{i \in F^{\star}_I} c^I_i + \sum_A p_A \sum_{i \in F^{\star}_A} c^A_i \leq B
$$
since $F^{\star}$ is feasible.
\end{proof}

\begin{theorem}\label{s6-valid}
If the original instance $\mathfrak I$ is feasible, then the solution returned by Algorithm~\ref{alg-5} has $\maxdist(F, A) \leq \rho+2$ for all scenarios $A \in Q$
\end{theorem}
\begin{proof}
By Lemma \ref{mu-feas}, if the given instance of $\mathcal{P}$-\textbf{poly} is feasible, then by specification of $\text{AlgRW}$ we get a feasible solution $F$. Consider now some $j \in A$ for $A \in Q$. The distance of $j$ to its closest facility is at most $d(\pi^A j, F_I\cup F_A) + d(j, \pi^A j)$. Since $\pi^A j \in H_A$, it will either have a stage-I open facility within distance $\rho R_j$, or we perform a stage-II opening in $G_{\pi^A(j)}$, which results in a covering distance of at most $R_j$. By the greedy clustering step, we have $R_{\pi^A j} \leq R_j$ and hence $d(j, \pi^A j) \leq R_j + R_{\pi^A j} \leq 2 R_j$. So $d(j, F_I \cup F_A) \leq (\rho+2)R_j$.
\end{proof}

\begin{theorem}
Algorithm~\ref{alg-5} is efficiently-generalizable, with $
|\mathcal S| \leq 2^m$.
\end{theorem}
\begin{proof} 
Given a newly arriving scenario $A$, we can set $(H_A, \pi^A) \gets $GreedyCluster$(A, R, -R)$, and then open the set $F^{\bar s}_A \gets \{ i^A_{j} ~|~ j \in H_A \text{ and } d(j,F_I) > \rho R_j\}$. Since this mimics the stage-II actions of Algorithm \ref{alg-5}, it satisfies property \ref{ef-gen-1}. By arguments of Theorem \ref{s6-valid}, we have $d(j,F_I \cup F_A) \leq (\rho+2)R_j$ for every $j \in A \in \mathcal{D}$, thus guaranteeing property \ref{ef-gen-2}. Finally, observe that the returned strategy depends solely on the set $F_I$, which has $2^m$ choices. Thus, $|\mathcal S| \leq 2^m$ as required in property \ref{ef-gen-3}.
\end{proof}

\subsection{Approximation Algorithm for Homogeneous \textbf{2S-MuSup} and \textbf{2S-MatSup}} \label{sec:homogmumat}
As an immediate consequence of Theorem~\ref{aga-thm0}, we can solve homogeneous \textbf{2S-MuSup} and \textbf{2S-MatSup} instances via solving the corresponding \textbf{RW-Sup} problems.

\begin{theorem}
\label{aga-thm}
There is a polynomial-time $3$-approximation for homogeneous \textbf{RW-MatSup}. There is a $3$-approximation algorithm for \textbf{RW-MuSup}, with runtime  $\poly(n,m,\Lambda)$. 
\end{theorem}

Combined with Theorem~\ref{aga-thm0}, this immediately gives the following results:
\begin{theorem}
\label{aga-thm2}
There is an efficiently-generalizable $5$-approximation for homogeneous \textbf{2S-MatSup-Poly}. There is an efficiently-generalizable $5$-approximation algorithm for \textbf{2S-MuSup-Poly} instances polynomially-bounded $\Lambda$
\end{theorem}

The algorithms for Theorem~\ref{aga-thm} are both very similar and are based on a solve-or-cut method of \cite{Chakra}. In either of the two settings, consider the following LP:
\begin{align}
&\sum_{i \in \mathcal{F}}y_i w_i + \sum_{j \in \mathcal C} x_j v_j \leq V \label{sa3-LP-0} \\
& x_j = \sum_{S \in \mathcal M, G_j \cap S = \emptyset} z_S ~~ \forall j \in \mathcal C \label{sa3-LP-1} \\
&y_i = \sum_{S \in \mathcal M, i \in S} z_S ~~\forall i \in \mathcal F \label{sa3-LP-3}\\
& 1 = \sum_{S \in \mathcal M} z_S \label{sa3-LP-4} \\
&0 \leq y_i, x_j, z_S \leq 1  \label{sa3-LP-5}
\end{align}

If the original problem instance is feasible, this LP instance is also clearly feasible. Although this LP has exponentially many variables (one for each feasible solution $S$), we will only maintain values for the variables $x_j, y_i$; we treat constraints (\ref{sa3-LP-1}, \ref{sa3-LP-3}, \ref{sa3-LP-4}) as implicit, i.e. $x, y$ should be in the convex hull defined by the $z$ constraints.

We will apply the Ellipsoid Algorithm to solve the LP. At each stage, we need to find a hyperplane violated by a given putative solution $x^*, y^*$ (if any). We assume (\ref{sa3-LP-0}) holds as otherwise we immediately have found our violated hyperplane. Let us set $(H, \pi) \leftarrow \text{GreedyCluster}(\mathcal C, R, x^*)$, and define $t_{j} = \sum_{j': \pi j' = j} v_{j'}$ for each $j \in H$. Now, for any solution $S \in \mathcal M$, let us define the quantity 
$$
\Psi(S) = \sum_{j \in H} \bigl( w(S \cap G_j) + \max\{0, 1 - |S \cap G_j| \} t_j \bigr)
$$

As we describe next, it is possible to choose $S \in \mathcal M$ to minimize $\Psi(S)$; let us put this aside for the moment, and suppose we have obtained such a solution $S$. If $\Psi(S) \leq V$, then we claim that $S$ is our desired 3-approximate solution. For, consider any client $j'$ with $\pi j' = j$. If $i \in S \cap G_j$, then $d(j', S) \leq d(j', j) + d(j, i) \leq 2 R + R = 3 R$. Thus, we have
$$
\sum_{i \in S} w_i + \sum_{j' \in \mathcal C: d(j', S) > 3 R} v_{j'} \leq 
\sum_{i \in S} w_i + \sum_{j' \in \mathcal C: |S \cap G_{\pi j}| = 0} v_{j'} = \Psi(S) \leq V
$$
as desired. Otherwise, suppose that no such $S$ exists, i.e. $\Psi(S) > V$ for all $S \in \mathcal M$. In particular, for any vector $z$ satisfying (\ref{sa3-LP-3}) and (\ref{sa3-LP-4}), along with corresponding vectors $x,y$, we would have
\begin{align*}
\sum_{j \in H} \Bigl( \sum_{i \in G_j} y_i w_i + x_j t_j \Bigr) &= \sum_{j \in H} \Bigl( \sum_{i \in G_j} \sum_{S \ni i}  z_S w_i + \sum_{S: G_j \cap S = \emptyset} z_S t_j \Bigr) \\
&= \sum_S z_S  \sum_{j \in H} \Bigl(  \sum_{i \in G_j \cap S} w_i + \max\{0, 1 - |G_j \cap S| \} t_j \Bigr) \\
&= \sum_S z_S \Psi(S) > \sum_S z_S V = V
\end{align*}

On the other hand, because of our greedy clustering, the given vectors $x^*, y^*$ satisfy:
\begin{align*}
\sum_{j \in H} \Bigl( \sum_{i \in G_j} y^*_i w_i + x_j^* t_j \Bigr) &= \sum_{j \in H} \Bigl( \sum_{i \in G_j} y^*_i w_i + x_j^* \sum_{j': \pi j' = j} v_{j'} \Bigr)
= \sum_{i \in \mathcal F} y^*_i w_i + \sum_{j \in \mathcal C}  v_{j} x^*_{\pi j} \\
&\leq \sum_{i \in \mathcal F} y^*_i w_i + \sum_{j \in \mathcal C}  v_{j} x^*_{j} \leq V
\end{align*}
where the last inequality holds by (\ref{sa3-LP-0}).  Thus, the hyperplane $\sum_{j \in H} ( \sum_{i \in G_j} y_i w_i + x_j t_j ) > V$ is violated by the solution $x^*, y^*$, and we can continue the Ellipsoid Algorithm.

To finish, we need to describe how to minimize $\Psi(S)$.  
\begin{proposition}
For a multi-knapsack $\mathcal M$, there is an algorithm to minimize $\Psi(S)$ with runtime $\poly(m,n,\Lambda)$.
\end{proposition}
\begin{proof}
Let $H = \{j_1, \dots, j_k \}$, and use a dynamic-programming approach to iterate through $H$: process the values $j_r$ for $r = 1, \dots, k$ while maintaining a table, indexed by the possible subsums for each of the $L$ knapsack constraints, listing the minimum possible value of $\Psi(S)$ among all solutions $S \in \{j_1, \dots, j_r \}$ for $r \leq k$. The table size is at most $\Lambda$. At each step, to update the table, we need to consider the possible facility $i \in G_j$ to add to $S$.
\end{proof}

\begin{proposition}
For a matroid $\mathcal M$, there is a polynomial-time algorithm to minimize $\Psi(S)$.
\end{proposition}
\begin{proof}
We can assume that the optimal $S$ has $|S \cap G_j| \leq 1$ for each $j$.  Thus, we can view this as finding a minimum-weight independent set of the intersection of two matroids, namely $\mathcal M$ and a partition matroid defined by the constraint $|S \cap G_j| \leq 1$ for all $j$. (Note that sets $G_j$ are pairwise disjoint.) The weight of any item $i \in G_j$ is then $w_i - t_j$.
\end{proof}

\section{Approximation Algorithm for Inhomogeneous \normalfont{2S-MatSup}}
\label{sec:5-apprx}
For this, we use the generic reduction strategy of Section~\ref{sec:mu-apprx}, providing a 9-approximation for inhomogeneous \textbf{RW-MatSup}, inspired by a similar rounding algorithm of \cite{Harris2018} for $k$-supplier. We will show the following:
\begin{theorem}
There is a 9-approximation algorithm for inhomogeneous \textbf{RW-MatSup}. 
There is an efficiently-generalizable 11-approximation for 
inhomogeneous \textbf{2S-MatSup-Poly}, with $|\mathcal S| \leq 2^m$
\end{theorem}

By Theorem~\ref{aga-thm0}, the second result follows immediately from the first, so we only consider the \textbf{RW-MatSup} setting.  Let us fix some matroid $\mathcal M$, radius demands $R_j$,  and weights $w_i, v_j, V$, and we assume access to the matroid rank function $r_{\mathcal{M}}$.  Now consider LP (\ref{s3-LP-1})-(\ref{s3-LP-4}), which we call the \emph{Preliminary LP}
\begin{align}
& x_j \geq 1 - \sum_{i \in G_j} y_i ~~ \forall j \label{s3-LP-1} \\
&\sum_{i \in \mathcal{F}}y_i w_i + \sum_{j \in \mathcal C} x_j v_j \leq V \label{s3-LP-2} \\
&\displaystyle\sum_{i \in U}y_i \leq r_{\mathcal{M}}(U), ~~\forall U \subseteq \mathcal{F} \label{s3-LP-3}\\
&0 \leq y_i, x_j \leq 1  \label{s3-LP-4}
\end{align}

If the \textbf{RW-MatSup} instance is feasible, then so is this LP (set $y_i = 1$ if $i \in S$, and setting $x_j = 1$ if $j$ is an outlier).  Although it has an exponential number of constraints, it can be solved in polynomial time via the Ellipsoid algorithm \cite{Ravi}. 

Starting with solution $x,y$ to this LP, we use an iterative rounding strategy. This is based on maintaining client sets $C_0, C_1, C_s$; the clients in the first group will definitely be outliers, the clients in the second group will definitely be non-outliers, and the clients in the second group are not yet determined yet. During the iterative process, we preserve the following invariants on the sets:
\begin{enumerate}[label=\textbf{I\arabic*}]
    \item \label{inv-1} For all $j,j' \in C_1$, with $j \neq j'$, we have $G_j \cap G_{j'} = \emptyset$.
    \item \label{inv-2} $C_0, C_1, C_s$ are pairwise disjoint.
\end{enumerate}

The iterative rounding is based on the following LP (\ref{s4-LP-aux-2})-(\ref{s4-LP-aux-6}), which we call the \emph{Main LP}.
\begin{align}
\text{minimize } &\displaystyle\sum_{i \in \mathcal{F}}z_i \cdot w_i + \sum_{j \in C_0} v_j + \sum_{j \in C_s} (1 - z(G_j)) v_j \\
\text{subject to }
&z(G_j) \geq 1, &\forall j \in C_1 \label{s4-LP-aux-2}\\
&z(G_j) = 0, &\forall j \in C_0 \label{s4-LP-aux-3} \\
&z(G_j) \leq 1, &\forall j \in C_s \label{s4-LP-aux-4} \\
&z(U) \leq r_{\mathcal{M}}(U), &\forall U \subseteq \mathcal{F} \label{s4-LP-aux-5} \\
&0 \leq z_i \leq 1 &\forall i \in \mathcal{F} \label{s4-LP-aux-6}
\end{align}

\begin{lemma}\label{s4-lem-2}
Let $z^{*}$ be an optimal vertex solution of the Main LP when $C_s, C_0, C_1$ satisfy \ref{inv-1}, \ref{inv-2}. Then if $C_s \neq \emptyset$, there exists at least one $j \in C_s$ with $z^{*}(G_j) \in \{0,1\}$. Moreover, if $C_s = \emptyset$ then the solution is integral, i.e., for all $i \in \mathcal{F}$ we have $z^{*}_i \in \{0,1\}$.
\end{lemma}
\begin{proof}
Suppose that $z^*(G_j) \in (0,1)$ for all $j \in C_s$. In this case, all constraints (\ref{s4-LP-aux-4}) are slack. All the constraints (\ref{s4-LP-aux-2}) must be then be tight; if not, we could increment or decrement $z_i$ for $i \in G_j$ while preserving all constraints.  Thus, $z^*$ must be an extreme point of the following system:
\begin{align}
&z(G_j) = 1, &\forall j \in C_1 \label{s5-LP-aux-2}\\
&z(U) \leq r_{\mathcal{M}}(U), &\forall U \subseteq \mathcal{F} \label{s5-LP-aux-5} \\
&0 \leq z_i \leq 1 &\forall i \in \mathcal{F} \\
&0 = z_i & \forall i \in G_j, j \in C_0 
\label{s5-LP-aux-6}
\end{align}

But this is an intersection of two matroid polytopes (the polytope for $\mathcal M$ projected to elements $i$ with $z_i \neq 0$, and the partition matroid corresponding to the disjoint sets $G_j$ for $j \in C_1$), whose extreme points are integral.
\end{proof}

\begin{algorithm}[t]
\begin{algorithmic}
\STATE {Solve the Preliminary LP to get $x, y$}
\IF {no feasible LP solution} \RETURN INFEASIBLE \ENDIF
\STATE {$C_0 \gets \emptyset$}\;
\STATE {$(C_1, \pi_t) \gets \text{GreedyCluster}(\{j \in \mathcal{C} ~|~ y(G_j) > 1 \}, R, -R)$}
\STATE {$C_{s} \gets \{j \in \mathcal{C} ~|~ y(G_j) \leq 1 \text{ and } \forall j' \in C_1: (G_j \cap G_{j'} = \emptyset \vee R_j < R_{j'} / 2)\}$}
\WHILE {$C_s \neq \emptyset$}
\STATE {Solve the Main LP using the current $C_s, C_0, C_1$, and get a basic solution $z$}
\STATE {Find a $j \in C_s$ with $z(G_j) \in \{0,1\}$ and set
$C_s \gets C_s \setminus \{j\}$}
\IF {$z(G_j) = 0$}
\STATE {$C_0 \gets C_0 \cup \{j\}$}
\ELSE
\STATE {$C_1 \gets C_1 \cup \{j\}$}
\FOR {any client $j' \in C_1 \cup C_s$ with $G_{j} \cap G_{j'} \neq \emptyset$ and $R_{j'} \geq R_{j}/2$}
\STATE {$C_s \gets C_s \setminus \{ j' \}$, $C_1 \gets C_1 \setminus \{j ' \}$}
\ENDFOR
\ENDIF
\ENDWHILE
\STATE {Solve the Main LP once more using the current $C_s = \emptyset, C_0, C_1$, and get a solution $z^{\text{final}}$}
\RETURN $S = \{i: z^{\text{final}}_i = 1 \}$
\end{algorithmic}
\caption{Iterative Rounding for \textbf{RW-MatSup}}\label{alg-4}
\end{algorithm}

Algorithm \ref{alg-4} shows the main iterative rounding process. We use $z^{(h)}$ to denote the solution obtained in iteration $h$, and $C^{(h)}_s, C^{(h)}_0, C^{(h)}_1$ for the client sets at the end of the $h^{\text{th}}$ iteration. We also use $z^{(T+1)}$ for $z^{\text{final}}$,  where $T$ is the total number of iterations of the main while loop. We also let $C^{(0)}_s, C^{(0)}_0, C^{(0)}_1$ be the client sets before the start of the loop, and we write $C^{(h)}$ as shorthand for the ensemble $C^{(h)}_0, C^{(h)}_1, C^{(h)}_s$.

\begin{lemma}\label{s4-lem-3}
For every $h=0,1,\hdots T$, the vector $z^{(h)}$ and sets $C^{(h)}$ satisfy invariants \ref{inv-1} and \ref{inv-2} satisfy the Main LP with objective function value at most $V$.
\end{lemma}
\begin{proof}
We show this by induction on $h$. For the base case $h = 0$, consider the fractional solution $z = y$. This satisfies constraint (\ref{s4-LP-aux-2}), since the clients $j \in C^{(0)}_t$ have $y(G_j) > 1$. Constraint (\ref{s4-LP-aux-3}) holds vacuously since $C_0^{(0)} = \emptyset$.  Constraint (\ref{s4-LP-aux-4}) is satisfied, because all clients $j \in C^{(0)}_s$ have $y(G_j) \leq 1$.
Constraint (\ref{s4-LP-aux-5}) holds due to $y$ already satisfying (\ref{s3-LP-3}). The greedy clustering step for $C^{(0)}_t$ ensures property \ref{inv-1}. Since $C^{(0)}_s$ contains only clients $j$ with $y(G_j) \leq 1$, \ref{inv-2} also holds.

Finally, for the objective value, we have
$$
\sum_{i \in \mathcal F} z_i w_i + \sum_{j \in C_0^{(0)}} v_j + \sum_{j \in C_s^{(0)}} (1 - z(G_j)) v_j = \sum_{i \in \mathcal F} z_i w_i  + \sum_{j \in C_s^{(0)}} (1 - y(G_j)) v_j =  \sum_{i \in \mathcal F} z_i w_i  + \sum_{j \in C_s^{(0)}} x_j v_j \leq V
$$
where here the penultimate inequality holds since, for $j \in C_s^{(0)}$, we have $y(G_j) \leq 1$ and hence $x_j = 1 - y(G_j) = 1 - z(G_j)$

Consider the induction step $h > 0$. By the inductive hypothesis the solution $z^{(h-1)}$ is feasible for the Main LP defined using sets $C^{(h-1)}_0, C^{(h-1)}_1, C^{(h-1)}_s$. By Lemma~\ref{s4-lem-2}, we can find an optimal vertex solution  $z^{(h)}$ with $z^{(h)}_i(G_j) \in \{0,1 \}$ for some $j_h \in C^{(h-1)}_s$; this new solution can only decrease the objective value. We need to ensure that $z^{(h)}$ remains feasible and the objective value does not increase after updating $C_s, C_0, C_1$, and these new sets satisfy the proper invariants.

If $z^{(h)}(G_{j_h}) = 0$, then $C^{(h)}_s = C^{(h-1)}_s \setminus \{j_h\}$ and $C_0^{(h)} = C_0^{(h)} \cup \{j_h \}$, and we need to verify (\ref{s4-LP-aux-3}) for $j_h$; this holds since $z^{(h)}(G_{j_h}) = 1$. 
The objective value does not change since previously $j_h \in C_s^{(h-1)}$ and $1 - z(G_{j_h}) = 1$.
Similarly, suppose $z^{(h)}(G_{j_h}) = 1$, so that $C^{(h)}_s = C^{(h-1)}_s \setminus \{j_h\}$, and $C^{(h)}_t = C^{(h-1)}_t \cup \{j_h\}$. Here we only need to verify (\ref{s4-LP-aux-2}) for $j_h$, but this holds since $z^{(h)}(G_{j_h}) = 1$. 

\ref{inv-2} remains true because we just moved $j_h$ from $C^{(h-1)}_s$ to either $C^{(h-1)}_0, C^{(h-1)}_1$ or discarded it. \ref{inv-1} can only be violated if $j_h\in C^{(h)}_1$, and there was a $j \in C^{(h-1)}_1$ with $G_j \cap G_{j_h} \neq \emptyset$ and $R_j < \frac{R_{j_h}}{2}$. However, this is impossible, because when $j$ first entered $C_1^{(h')}$ at some earlier time $h'$ (possibly $h' = 0$) it should have removed $j_h$ from $C_s$.
\end{proof}

In particular, Lemma \ref{s4-lem-2} at each iteration $h$ ensures that we can always find a $j \in C_s$ with $z^{(h)}(G_j) \in \{0,1\}$, and remove it from $C_s^{(h+1)}$. Thus the loop must terminate after at most $n$ iterations, and the final values $T$,  $z^{\text{final}}$, and so on are well-defined.

\begin{lemma}\label{s4-lem-5}
If $j \in C^{(h)}_1$ for any $h$, then $d(j, S) \leq 3R_j$. 
\end{lemma}
\begin{proof}
We show this by induction backward on $h$. When $h=T+1$ this is clear, since then $z(G_j) = 1$, so we will open a facility in $G_j$ and have $d(j, S) \leq R_j$.  Otherwise, suppose that $j$ was removed from $C_1$ in iteration $h$. This occurs because $j_{h}$, the client chosen in iteration $h$, entered $C_1$. Therefore, $G_j \cap G_{j_{h}} \neq \emptyset$. Moreover, because $j_{h} \in C^{(0)}_s$ and also $j_h$ was not removed from $C_s$ when $j$ first entered $C_1$, we have $R_{j_{h}} \leq R_j / 2$. Finally, since $j_{h}$ is present in $C^{(h)}_1$, the inductive hypothesis applied to $j_h$ gives $d(j_{h}, S \leq 3R_{j_{h}}$. Overall we get $d(j, S) \leq R_j + R_{j_{h}} + d(j_{h},S) \leq 3R_j$.
\end{proof}

\begin{lemma}\label{s4-lem-6}
If $j \notin C^{\text{final}}_0$, we have $d(j, S) 
\leq 9 R_j$.
\end{lemma}
\begin{proof}
First suppose $y(G_j) > 1$. In this case, the greedy clustering step to form $C^{(0)}_1$ gives $G_j \cap G_{j'} \neq \emptyset$ for some $j' \in C^{(0)}_1$ with $R_{j'} \leq R_j$. Furthermore, Lemma~\ref{s4-lem-5} gives $d(j', S) \leq 3 R_{j'}$. Overall, $d(j, S) \leq d(j, j') + d(j', S) \leq R_j + R_{j'} + 3 R_{j'} \leq 5 R_j$. 

Next, suppose that $j \in C^{(0)}_s$, but $j$ was later removed from $C_s^{(h)}$ for $h > 0$. If $j$ was moved into $C_1^{(h)}$, then Lemma~\ref{s4-lem-5} would give $d(j, S) \leq 3 R_j$. If $j$ was moved into $C_0^{(h)}$, it remains in $C^{\text{final}}_0$ and there is nothing to show. Finally, suppose $j$ was removed from $C_s^{(h)}$ because $G_j \cap G_{j_h} \neq \emptyset$ and $R_j \geq R_{j_h}/2$, where client $j_h$ entered $C_1^{(h)}$. By Lemma~\ref{s4-lem-5}, we have $d(j_h, S) \leq 3 R_{j_h}$.
So $d(j, S) \leq d(j, j_h) + d(j_h, S) \leq R_j + R_{j_h} + 3 R_{j_h} = R_j + 4 R_{j_h}$. Since $R_{j_h} \leq 2 R_j$, this is at most $9 R_j$. 

The case where $y(G_j) \leq 1$ but $j \notin C_s^{(0)}$ is completely analogous: there must be some $j' \in C_1^{(0)}$ with $G_j \cap G_{j'} \neq \emptyset$ and $R_j \geq R_{j'}/2$, and again $d(j', S) \leq 3 R_{j'}$ and $d(j,S) \leq 9 R_j$.
\end{proof}

\begin{theorem}\label{s4-lem-4a}
The set $S$ returned by Algorithm~\ref{alg-4} satisfies the matroid constraint, and is a 9-approximation for the budget constraint.
\end{theorem}
\begin{proof}
When the algorithm terminates with $C_s = \emptyset$, Lemma~\ref{s4-lem-2} ensure the solution $z^{\text{final}}$ is integral. By Lemma~\ref{s4-lem-6}, any client $j$ with $d(j, S) > 9 R_j$ must have $j \in C^{\text{final}}_0$. Hence, $\sum_{j: d(j, S) > 9 R_j} v_j \leq \sum_{j \in C_0} v_j$. For the facility costs, we have $\sum_{i \in S} w_i = \sum_i z_i^{\text{final}} w_i$. Finally, by Lemma~\ref{s4-lem-3}, and noting that $C_s^{\text{final}} = \emptyset$, we have $\sum_{i} z_i^{\text{final}} w_i + \sum_{j \in C_0} v_j \leq V$.
\end{proof}

\section{Acknowledgements}

A preliminary version of this work appeared in the \emph{Proc.\ International Workshop on Approximation Algorithms for Combinatorial Optimization Problems (APPROX), 2021}.

Brian Brubach was supported in part by NSF awards CCF-1422569 and CCF-1749864, and by research awards from Adobe. Nathaniel Grammel and Leonidas Tsepenekas were supported in part by NSF awards CCF-1749864 and CCF-1918749, and by research awards from Amazon and Google. Aravind Srinivasan was supported in part by NSF awards CCF-1422569, CCF-1749864 and CCF-1918749, and by research awards from Adobe, Amazon, and Google.

The authors sincerely thank Chaitanya Swamy as well as the referees of earlier versions of this paper, for their precious feedback and helpful suggestions.

\printbibliography

\appendix
\section{Standard SAA Methods in the Context of Supplier Problems}\label{sec:appendix}
Consider the general two-stage stochastic setting: in the first stage, we take some proactive action $x \in X$ (i.e., the stage-I openings), incurring cost $c(x)$, and in the second stage, a scenario $A$ is sampled from the distribution $\mathcal{D}$, and we take some \emph{stage-II} recourse actions $y_A \in Y_A$ (i.e., stage-II openings) with cost $f_A(x, y_A)$. The goal is to find a solution $x^\star \in X$ minimizing $f(x) = c(x) + \mathbb{E}_{A \sim \mathcal{D}}[q_A(x)]$, where $q_A(x) = \min_{y \in Y_A}\{f_A(x,y) \}$. In our supplier problems, $X = Y_{A} = 2^{\mathcal{F}}$, where $x$ corresponds to a set of stage-I openings $F_{I}$ and $y_{A}$ corresponds to the set of stage-II openings $F_{A}$ under scenario $A$. Then, $c(x) = \sum_{i\in x} c_{i}^{I}$ is the total stage-I cost of the facilities in $F_{I}$ and $f_{A}(x, y_{A}) = c(x) + \sum_{i \in y_{A}}c_{i}^{A}$ is the total cost incurred by the openings across both stages.

\subparagraph*{The Standard SAA Method:}Consider minimizing $f(x)$ in the black-box model. If $S$ is a set of scenarios sampled from the black-box oracle, let $\hat{f}(x) = c(x) + \big{(}\sum_{A \in S}q_A(x)\big{)}/|S|$ be the empirical estimate of $f(x)$. Also, let $x^*$ and $\bar x$ be the minimizers of $f(x)$ and  $\hat{f}(x)$ respectively.

If $f(x)$ is modeled as a convex or integer program, then for any $\epsilon, \gamma \in (0,1)$ and with $|S| = \poly(n,m, \lambda, \epsilon, 1 / \gamma)$, we have $f(\bar x) \leq (1+\epsilon)f(x^*)$ with probability at least $1-\gamma$ (where $\lambda$ is the maximum multiplicative factor by which an element's cost is increased in stage-II) \cite{swamySAA,Charikar2005}. If $\bar x$ is an $\alpha$-approximate minimizer of $\hat f(x)$, then a slight modification to the sampling (see \cite{Charikar2005}) still gives $f(\bar x) \leq (\alpha +\epsilon) f(x^*)$ with probability at least $1-\gamma$.

The result of \cite{Charikar2005} allows us to reduce the black-box model to the polynomial-scenarios model, as follows. First find an $\alpha$-approximate minimizer $\bar{x}$ of $\hat{f}(x)$, and treat $\bar{x}$ as the stage-I actions. Then, given any arriving $A$, re-solve the problem using any known $\rho$-approximation algorithm for the non-stochastic counterpart, with $\bar{x}$ as a fixed part of the solution. This process eventually leads to an overall approximation ratio of $\alpha \rho + \epsilon$.

\subparagraph*{Roadblocks for Supplier Problems:}
A natural way to fit our models within this framework would be to guess  the optimal radius $R^*$ and use the opening cost as the objective function $f_{R^*}(x)$, with the radius requirement treated as a hard constraint. In other words, we set $f_{R^*}(x) = c^I(x) + \mathbb{E}_{A \sim \mathcal{D}}[q_{A,R^*}(x)]$ with $q_{A,R^*}(x)=\min_y\{c^A(y) ~|~ (x,y) \text{ covers all } j \in A \text{ within distance } R^* \}$. Here $f_{R^*}(x)$ may represent the underlying convex or integer program. If the empirical minimizer $\bar x_{R^*}$ can be converted into a solution with coverage radius $\alpha R^*$ and opening cost at most $f_{R^*}(\bar x_{R^*})$, we get the desired result because $f_{R^*}(\bar x_{R^*}) \leq (1+\epsilon)f_{R^*}(x^*_{R^*})$ and $f_{R^*}(x^*_{R^*}) \leq B$. 

With slight modifications, all our polynomial-scenarios algorithms can be interpreted as such rounding procedures.  Nonetheless, we still have to identify a good guess for $R^*$, \textbf{and this is an unavoidable roadblock in the standard SAA for supplier problems.} Observe that $R$ is a good guess if $f_R(x^*_R) \leq (1+\epsilon)B$, since in this way vanilla SAA combined with our rounding procedures would give opening cost $f_R(\bar x_R) \leq (1+2\epsilon)f_R(x^*_R)$, and minimizing over the radius is just a matter of finding the minimum good guess. However, empirically estimating $f_R(x)$ within an $(1+\epsilon)$ factor, which is needed to test $R$, may require a super-polynomial number of samples \cite{kleywegt} due to scenarios with high stage-II cost and small probability. \textbf{On a high level, the obstacle in supplier problems stems from the need to not only find a minimizer $\bar{x}_R$, but also compute its corresponding value $f_R(\bar{x}_R)$}. This makes it impossible to know which guesses $R$ are good, and consequently there is no way to optimize over the radius.

(If the stage-II cost of every scenario is polynomially bounded, then the variance of $\hat f_R(x)$ is also polynomial, and standard SAA arguments go through without difficulties as in Theorem~\ref{thm13}.)

\end{document}